\documentclass[letterpaper, 10pt, conference]{IEEEtran} 
\usepackage{blindtext, graphicx}
\usepackage{amsmath} 
\usepackage{amssymb}
\usepackage[utf8]{inputenc}
\usepackage{graphicx}

\usepackage{lipsum}
\makeatletter

\usepackage{cite}

\ifCLASSINFOpdf
\else
\fi

\usepackage{lipsum,multicol}
\usepackage{subfigure}
\usepackage{floatrow}
\usepackage{algorithm}
\usepackage{listings}
\usepackage{mdwmath}
\usepackage{amsfonts}
\usepackage{algpseudocode}
\usepackage{amsmath}
\usepackage{diagbox}
\usepackage{caption}
\usepackage{amsthm}

\DeclareMathOperator*{\argmin}{argmin}

\DeclareMathOperator{\diag}{diag}

\theoremstyle{plain}
\newtheorem{theorem}{Theorem}
\newtheorem{lemma}[theorem]{Lemma}
\usepackage[bottom]{footmisc}

\newcommand{\CLASSINPUTtoptextmargin}{0.75in}
\newcommand{\CLASSINPUTbottomtextmargin}{1in}

\begin{document}
%
\title{Fronthaul Compression for Uplink Massive MIMO using Matrix Decomposition}
\vspace{0.05cm}
\author{\IEEEauthorblockN{Aswathylakshmi P and Radha Krishna Ganti\\}
    \IEEEauthorblockA{Department of Electrical Engineering\\
                      Indian Institute of Technology Madras\\
                      Chennai, India 600036\\
                      \{aswathylakshmi, rganti\}@ee.iitm.ac.in}
}



%


\maketitle

\begin{abstract}
Massive MIMO opens up attractive possibilities for next generation wireless systems with its large number of antennas offering spatial diversity and multiplexing gain. However, the fronthaul link that connects a massive MIMO Remote Radio Head (RRH) and carries IQ samples to the Baseband Unit (BBU) of the base station can throttle the network capacity/speed if appropriate data compression techniques are not applied. In this paper, we propose an iterative technique for fronthaul load reduction in the uplink for massive MIMO systems that utilizes the convolution structure of the received signals. We use an alternating minimisation algorithm for blind deconvolution of the received data matrix that provides compression ratios of 30-50. In addition, the technique presented here can be used for blind decoding of OFDM signals in massive MIMO systems. 
\end{abstract}


%
\IEEEpeerreviewmaketitle
\section{Introduction}
Massive MIMO is a pivotal technology in the 5G wireless standards \cite{larsson2017massive} and hopefully 6G standards. While it can dramatically increase the network capacity and speed with spatial diversity and multiplexing gain, it also comes with the practical difficulties of handling and processing large amounts of data by the hardware and the fronthaul network. With 5G NR standards envisioning base stations with antennas numbering in the range of 30-100 \cite{larsson2017massive}, one significant bottleneck is the fronthaul, which handles the data transfer between the remote radio head (RRH) and the baseband unit (BBU) at the radio base station. The latest Open Radio Access Network (O-RAN) standard used for fronthaul data transfer is the Enhanced CPRI (eCPRI), which allows for various functional splits between the BBU and RRH to flexibly reduce fronthaul load \cite{ecpri2017}. However, a low-PHY functional split requires a data rate upto 236 Gbps for 100MHz bandwidth for a 64-antenna base station \cite{ecpri2017}. Since the fronthaul load scales up with the number of RRH antennas, laying high speed optical fibres for increasing antenna elements will increase the CAPEX for network operators significantly and hence can limit the use of massive MIMO solutions.

\par{Numerous compression techniques have been proposed to reduce the fronthaul load for the uplink, which exploit either the redundancies in the received waveform \cite{guo2013lte}, \cite{drvenica2016}, a priori knowledge of the characteristics of the received signal \cite{Peng2016} or the correlation between the base station antennas \cite{Choi2016}, \cite{aswathylakshmi2019qr}. \cite{guo2013lte} achieves 1/2-rate compression by removing redundant spectrum bandwidth and compressing the quantization bit-width. \cite{drvenica2016} outlines a lossy compression algorithm by applying FFT and Discrete Cosine Transform (DCT) to the received signals and discarding low power frequency coefficients. While \cite{guo2013lte}, \cite{drvenica2016} have the advantage of low implementation complexity, they do not offer high compression ratios. \cite{Peng2016} reviews compressive sensing techniques that make
use of the sparsity of signals for uplink fronthaul compression. \cite{Choi2016} performs a principal component analysis (PCA) on the matrix of received signals to utilize the spatio-temporal correlation of the signals for their compression. \cite{aswathylakshmi2019qr} uses a low-rank approximation of the received signal matrix via QR decomposition to reduce the fronthaul load. Although \cite{Choi2016}, \cite{aswathylakshmi2019qr} achieve better compression than \cite{guo2013lte}, \cite{drvenica2016} at the cost of increased computations, the highest compression ratio they achieve for 64 antennas is 5. We propose a new method that offers a compression ratio nearly an order of magnitude higher than \cite{Choi2016}, \cite{aswathylakshmi2019qr} by leveraging the underlying convolution structure of the received data.}

\begin{figure*}[h!]
    \centering
    \captionsetup{justification=centering}
    \includegraphics[scale=0.4]{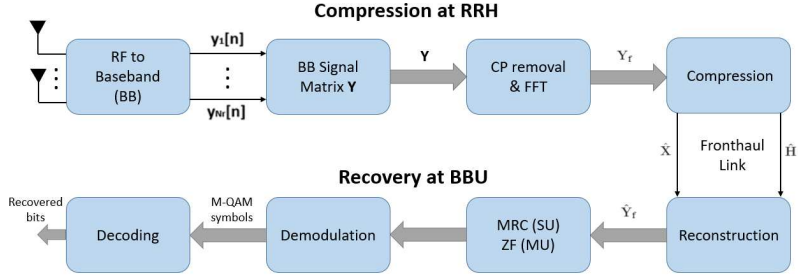}
    \caption{Proposed Fronthaul Compression Scheme}
    \label{fig:my_label}
\end{figure*}

\par{The goal of fronthaul compression is to compress $N$ IQ samples received at the $N_r$ RRH antennas, which can be arranged in the form of an $N\times N_r$ matrix. For example, $N$ can be the FFT size used in OFDM signals. In a massive MIMO system, $N_r$ is usually much higher than the number of users served in the system simultaneously. This means that the columns of the received signal matrix at RRH are the result of the same user data sequence convolved with a different multi-path channel response corresponding to each receive antenna in the time domain. In the frequency domain, this received signal matrix can be expressed as the product of an $N$-dimensional diagonal matrix of user data and the Fourier transform of an $L\times N_r$ channel matrix, where $L$ is the number of significant multi-paths. If $L \ll N_r$, as is typical for massive MIMO, representing the received signal matrix in terms of the $N$-length user data sequence and the $L\times N_r$ channel matrix allows us to recover it at the BBU using far fewer samples than $NN_r$. Towards this end, we need to perform a blind deconvolution of the received signal matrix at the RRH. Our objective is similar to the one in \cite{li2019rapid}, however we take an approach different from \cite{li2019rapid} to solve the objective function since we have to deconvolve a matrix of sequences rather than a single sequence as in \cite{li2019rapid}. Noting that the $N$-point Fourier transform of the $L\times N_r$ channel matrix is low rank, given $L< \min \{N,N_r\}$, we present an iterative algorithm consisting of alternating minimisation similar to the one in \cite{jain2013low} used for low rank matrix sensing. However, our algorithm is constructed differently from \cite{jain2013low} to reflect the fact that only one of the deconvolved matrices (the channel matrix) is low rank while the other (the user data matrix) is a full rank diagonal matrix.}

\par{This paper is organised as follows: Section II outlines our system model, Section III presents our compression algorithm for single user and multi-user cases, Section IV presents the results of link level simulations of the algorithm and analyses its performance vis-\'a-vis the PCA compression in \cite{Choi2016} and the uncompressed system, and Section V concludes the work.}

\section{System Model}

We consider a massive MIMO 5G base station RRH that has $N_r$ antennas and receives signals from $N_u$ users in the uplink. Since the uplink multiple access scheme in 5G (and 4G) is OFDMA \cite{38101}, the bit-stream from each user is mapped to an M-QAM symbol constellation followed by sub-carrier mapping, IFFT, and cyclic prefix (CP) addition. We assume that the channel has a maximum of $L$ significant multi-paths for each user. The signal received at antenna $r$ at a sampling instant $n$ is
\begin{equation*}
y_{r}[n] = \sum_{u=1}^{N_u} x_{u}[n] \circledast h_{r,u}[n] + w_{r}[n] ,  
\end{equation*}
where $x_{u}[n]$ is the OFDM symbol of user $u$, $h_{r,u}[n]$ is the multi-path channel response between user $u$ and antenna $r$, $\circledast$ represents circular convolution (because of cyclic prefix in OFDM), and $w_r$ is the additive white Gaussian noise (AWGN) at antenna $r$ with variance $\sigma^2$. 
\par{We consider a block of $N$ samples in the time-domain received at the RRH after CP removal. We assume a constant channel for the duration of these $N$ samples. Without loss of generality, we choose $N$ as the FFT size of one OFDM symbol. The signal matrix received at the RRH is
\begin{equation*}
\mathbf{Y}=
  \begin{bmatrix}
   y_{1}[1] & y_{2}[1] & . & . & . & y_{N_r}[1] \\
    y_{1}[2] & y_{2}[2] & . & . & . & y_{N_r}[2]\\
    . & . & .  &   &   &. \\
    . & . &   &  . &   &. \\
    . & . &   &   & .  &. \\
    y_{1}[N] & y_{2}[N] & . & . & . & y_{N_r}[N]\\
  \end{bmatrix}_{N\times N_r} .
\end{equation*}
Here, each column of $\mathbf{Y}$ represents the received signal at each antenna over $N$ sampling instants. We need to send information about $\mathbf{Y}$ from the RRH to the BBU via the fronthaul link for its faithful reconstruction at the BBU.}

\section{Compression using Matrix Decomposition}
For ease of analysis, we first focus on the scenario where all the $N$ sub-carriers in the uplink are occupied by a single user and describe how to compress $\mathbf{Y}$. We then extend this result to the case where the $N$ sub-carriers are occupied by multiple users simultaneously. Fig.1 gives the outline of the process of compression at RRH and recovery at BBU.

\subsection{SU-MIMO System}
When the number of users, $N_u$ = 1, we have a single-user multiple-input multiple-output (SU-MIMO) system in the  uplink. Applying FFT to $\mathbf{Y}$ and using the subscript $f$ to denote quantities in the frequency domain, we can write $\mathbf{Y}$ in the frequency domain as
\begin{equation}
    \mathbf{Y_f}= \mathbf{X_f}\mathbf{H_f} + \mathbf{W_f},
\end{equation}
where $\mathbf{X_f}$ is the $N\times N$ diagonal matrix with the $N$-length M-QAM user data as its diagonal, $\mathbf{H_f}$ is the $N\times N_r$ multi-path channel matrix, and $\mathbf{W_f}$ is the noise. Denoting the first $L$ columns of the $N\times N$ DFT matrix as $\mathbf{F_L}$, $\mathbf{H_f}$ can be expressed as the matrix product $\mathbf{{F_{L}H_t}}$, where $\mathbf{H_t}$ is the $L\times N_r$ time-domain multi-path channel response for the $N_r$ antennas. Using this, we rewrite (1) as 
\begin{equation}
    \mathbf{Y_f}= \mathbf{X_f}\mathbf{{F_{L}H_t}} + \mathbf{W_f}.
\end{equation}

\par{If $\mathbf{Y_f}$ is transmitted on the fronthaul, then $N N_{r}$ samples are required to be sent. However, if we observe (2), we see that we can reconstruct the signal component of $\mathbf{Y_f}$ with the $N$ non-zero samples corresponding to $\mathbf{X_f}$ and the $LN_r$ samples corresponding to $\mathbf{H_t}$. Therefore, if we decompose $\mathbf{Y_f}$ into its signal components $\mathbf{X_f}$ and $\mathbf{H_t}$,  we need to send only $N + LN_r$ samples from the RRH to the BBU, to recover $\mathbf{Y_f}$. Typically, $L \ll \min \{N_r, N\}$, therefore this gives a compression ratio (CR) of}
\begin{equation}
    CR_{SU} = \frac{N N_{r}}{N+LN_r}.
\end{equation}

\par{One way of obtaining $\mathbf{X_f}$ and $\mathbf{H_t}$ from $\mathbf{Y_f}$ is coherent demodulation of $\mathbf{Y_f}$. This would require pilots and modulation order to be known at the RRH \cite{zaidi20185g}. However, in general, this information is not available at the RRH \cite{alliance2020ran}. The idea is to decompose $\mathbf{Y_f}$ into matrices $\mathbf{\hat{X}}$ and $\mathbf{\hat{H}}$ (not necessarily equal to $\mathbf{X_f}$ and $\mathbf{H_t}$) so that $\mathbf{Y_f} \approx \mathbf{\hat{X}F_{L}\hat{H}}$, without any knowledge of pilots/modulation order. This clearly resembles the blind demodulation of massive MIMO OFDM symbols and the technique given in this paper can also be used for the same. Thus, for compressing $\mathbf{Y_f}$, we need to find a diagonal matrix $\hat{\mathbf{X}}$ and an $L\times N_r$ matrix $\hat{\mathbf{H}}$ that minimise
$|| \mathbf{Y_f} - \hat{\mathbf{X}}\mathbf{{F_{L}\hat{H}}} ||^{2}_{F}$. We use an iterative alternating minimisation approach to find $\hat{\mathbf{X}}$ and $\hat{\mathbf{H}}$ as outlined in Algorithm 1. At each iteration of the algorithm, we solve the minimisation problems
\begin{equation}
\begin{aligned}
    \mathbf{\hat{H}_{k+1}} &= \argmin_{\mathbf{H}} || \mathbf{Y_f} - \mathbf{\hat{X}_k}\mathbf{{F_{L}{H}}} ||^{2}_{F}, \\
    \mathbf{\hat{X}_{k+1}} &= \argmin_{\mathbf{X}} || \mathbf{Y_f} - \mathbf{X F_{L}\hat{H}_{k+1}} ||^{2}_{F},
\end{aligned}
\end{equation}
alternately until we obtain the optimal solution to faithfully reconstruct $\mathbf{Y_f}$ within a pre-defined error tolerance $\epsilon$. The above optimisation problem has closed form solution as shown in Algorithm 1.  

\par{We begin with an initial guess, $\mathbf{\hat{X}_0}$, and solve for $\mathbf{\hat{H}}$ in $\mathbf{Y_f} =  \mathbf{\hat{X}_0}\mathbf{{F_{L}{\hat{H}}}}$,
as given by step 5 of Algorithm 1.
We then use this $\mathbf{\hat{H}}$ to find $\hat{\mathbf{X}}$ such that $\sum_{r=1}^{N_r} || \mathbf{y_r} - \mathbf{\hat{X}F_{L}h_r} ||^{2}_{2}$ is minimised, where $\mathbf{y_r}$ and $\mathbf{\hat{h}_r}$ are the $r$-th columns of $\mathbf{Y_f}$ and $\mathbf{\hat{H}}$, respectively. This leads us to the solution given in step 8 of Algorithm 1. We repeat this process until the product $\mathbf{\hat{X}F_{L}\hat{H}}$ meets the error tolerance $\epsilon$. We then send only $\mathbf{\hat{X}}$ and $\mathbf{\hat{H}}$ to the BBU.}


\renewcommand\footnoterule{}      
\begin{algorithm}[b!] 
\caption{Alternating minimisation for SU-MIMO fronthaul compression}\footnotetext{$\dagger$ denotes matrix pseudo-inverse, $*$ denotes complex conjugation, $\mathbf{A_{m,r}}$ denotes element at row $m$ and column $r$ of matrix $\mathbf{A}$ and $x_{i}(m)$ denotes $m^{th}$ diagonal element of $\mathbf{\hat{X}_i}$.}
\label{alg:loop}
\begin{algorithmic}[1]
\State Input $\mathbf{Y_f}$
\State Define error tolerance $\epsilon$
\State Initialize diagonal of ${\mathbf{\hat{X}_0}}$ to be the first column of $\mathbf{Y_f}$, $k = 0$, $\mathbf{B} = [\hspace{1mm}]$
\Repeat
 \State {${\mathbf{\hat{H}_{k+1}}}\gets \mathbf{{(\hat{X}_{k}F_L)}^{\dagger}}\mathbf{Y_f}$}
 \State $\mathbf{B} \gets \mathbf{F_{L}\hat{H}_{k+1}}$
 \For {$m \gets 1$ to $N$}
    \State $x_{k+1}(m) = \Big(\sum_{r=1}^{N_r} \mathbf{Y_{m,r}B^{*}_{m,r}}\Big)/\Big(\sum_{r=1}^{N_r} |\mathbf{B_{m,r}}|^2\Big)$
 \EndFor
 \State $k \gets k+1$
\Until{$|| \mathbf{Y_f} - {\mathbf{\hat{X}_k}}\mathbf{{F_{L}\hat{H}_k}} ||_{F} / ||\mathbf{Y_f}||_{F} < \epsilon$}\\
\Return $\mathbf{\hat{X}_k}, \mathbf{\hat{H}_k}$
\end{algorithmic}
\end{algorithm}
}


\subsection{MU-MIMO System}
When the number of users occupying the same set of sub-carriers in the system, $N_{u} > 1$, we have a multi-user multiple-input multiple-output (MU-MIMO) system. For such a system, we can express $\mathbf{Y}$ in the frequency domain as
\begin{equation}
    \mathbf{Y_f}= [\mathbf{X_{f}(1)} \hspace{1mm}\mathbf{X_{f}(2)} \hspace{1mm}... \hspace{1mm}\mathbf{X_{f}(N_u)}] 
    \begin{bmatrix}
        \mathbf{H_{f}(1)} \\
        \mathbf{H_{f}(2)} \\
        \vdots \\
        \mathbf{H_{f}(N_u)} \\
    \end{bmatrix} + \mathbf{W_f},
\end{equation}
where $\mathbf{X_{f}(u)}$ is the $N \times N$ diagonal matrix, with the $N$-length M-QAM data of user $u$ as its diagonal, and $\mathbf{H_{f}(u)}$ is the $N\times N_r$ multi-path channel matrix for user $u$. Similar to the SU-MIMO case, we can express each $\mathbf{H_{f}(u)}$ as the matrix product $\mathbf{{F_{L}H_{t}(u)}}$, where $\mathbf{H_{t}(u)}$ is the $L\times N_r$ time-domain multi-path channel response for user $u$. Let $\diag(\mathbf{F_L})$ denote the block matrix of the form
\begin{equation*}
\diag(\mathbf{F_L})=
  \begin{bmatrix}
   \mathbf{F_{L}} & 0 & . & . & . & 0 \\
    0 & \mathbf{F_{L}} & . & . & . & 0\\
    . & . & .  &   &   &. \\
    . & . &   &  . &   &. \\
    . & . &   &   & .  &. \\
    0 & 0 & . & . & . & \mathbf{F_{L}}\\
  \end{bmatrix}_{N N_{u}\times L N_{u}} .
\end{equation*}
Then,
\begin{equation*}
    \begin{bmatrix}
        \mathbf{H_{f}(1)} \\
        \mathbf{H_{f}(2)} \\
        \vdots \\
        \mathbf{H_{f}(N_u)} \\
    \end{bmatrix}_{N N_{u} \times N_r} = \diag(\mathbf{F_L})\begin{bmatrix}
        \mathbf{H_{t}(1)} \\
        \mathbf{H_{t}(2)} \\
        \vdots \\
        \mathbf{H_{t}(N_u)} \\
    \end{bmatrix}_{L N_{u} \times N_r}.
\end{equation*}
Using the above, we can rewrite (5) as
\begin{equation}
    \mathbf{Y_f}= [\mathbf{X_{f}(1)} \hspace{1mm}\mathbf{X_{f}(2)} \hspace{1mm}... \hspace{1mm}\mathbf{X_{f}(N_u)}] \diag(\mathbf{F_L})
    \begin{bmatrix}
        \mathbf{H_{t}(1)} \\
        \mathbf{H_{t}(2)} \\
        \vdots \\
        \mathbf{H_{t}(N_u)} \\
    \end{bmatrix} + \mathbf{W_f}.
\end{equation}

\renewcommand\footnoterule{}      
\begin{algorithm}[b!] 
\caption{Alternating minimisation for MU-MIMO fronthaul compression}\footnotetext{$\mathbf{x^{T}_{i}(m)} = [x_{m}(1) \hspace{1mm} x_{m}(2) \hspace{1mm} ... \hspace{1mm}  x_{m}(N_u)]$, where $x_{m}(u)$ is the $m^{th}$ diagonal element of $\mathbf{\hat{X}_{i}(u)}$, $\mathbf{y_{m}^{T}}$ is $m^{th}$ row of $\mathbf{Y_f}$, $\mathbf{B_m}$ is the $N_{u} \times N_{r}$ sub-matrix of $\mathbf{B}$ for $m$-th sub-carrier for the $N_u$ users.}
\label{alg:loop}
\begin{algorithmic}[1]
\State Input $\mathbf{Y_f}$
\State Define error tolerance $\epsilon$
\State Initialize diagonals of ${\mathbf{\hat{X}_{0}(u)}}$ to be any $N_u$ columns of $\mathbf{Y_f}$ for $u = 1,2,...,N_u$, $k = 0$, $\mathbf{B} = [\hspace{1mm}]$
\Repeat
 \State {${\mathbf{\hat{H}_{k+1}}}\gets \mathbf{{(\hat{X}_{k}}\diag(\mathbf{F_L))}^{\dagger}}\mathbf{Y_f}$}
 \State $\mathbf{B} \gets \diag(\mathbf{F_{L})\hat{H}_{k+1}}$
 \For {$m \gets 1$ to $N$}
    \State $\mathbf{x^{T}_{k+1}(m)} = \mathbf{y_{m}^{T} (B_m)^{\dagger}}$
 \EndFor
 \State $k \gets k+1$
\Until{$|| \mathbf{Y_f} - {\mathbf{\hat{X}_k}}\diag(\mathbf{{F_{L})\hat{H}_k}} ||_{F} / ||\mathbf{Y_f}||_{F} < \epsilon$}\\
\Return $\mathbf{\hat{X}_k}, \mathbf{\hat{H}_k}$
\end{algorithmic}
\end{algorithm}

Let $\mathbf{\hat{X}} = [\mathbf{\hat{X}(1)} \hspace{1mm}\mathbf{\hat{X}(2)} \hdots \mathbf{\hat{X}(N_u)}]$, where $\mathbf{\hat{X}(u)}$ are $N \times N$ diagonal matrices, and $\mathbf{\hat{H}}= [\mathbf{\hat{H}(1)^{T}}\hspace{1mm}\mathbf{\hat{H}(2)^{T}} \hdots \mathbf{\hat{H}(N_u)^{T}}]^T$
where $\mathbf{\hat{H}(u)}$ are size $L \times N_r$ matrices, for $u = 1,2,.., N_u$. We use the same iterative approach to find $\mathbf{\hat{X}}$ and $\mathbf{\hat{H}}$ that form the signal component of $\mathbf{Y_f}$ as was used in the SU-MIMO case. As outlined in Algorithm 2, at each iteration we solve the minimisation problems
\begin{equation}
\begin{aligned}
    \mathbf{\hat{H}_{k+1}} &= \argmin_{\mathbf{H}} || \mathbf{Y_f} - \mathbf{\hat{X}_k}\diag(\mathbf{{F_{L}){H}}} ||^{2}_{F},\\
    \mathbf{\hat{X}_{k+1}} &= \argmin_{\mathbf{X}} || \mathbf{Y_f} - \mathbf{X} \diag(\mathbf{F_{L})\hat{H}_{k+1}} ||^{2}_{F},
\end{aligned}
\end{equation}
alternately until we obtain the optimal solution \{$\mathbf{\hat{X}, \hat{H}}$\} to faithfully reconstruct $\mathbf{Y_f}$ within a pre-defined error tolerance $\epsilon$. Then we send the $N N_u$ samples corresponding to $\mathbf{\hat{X}}$ and the $L N_{r} N_{u}$ samples corresponding to $\mathbf{\hat{H}}$ from the RRH to the BBU, leading to a compression ratio of
\begin{equation}
    CR_{MU} = \frac{N N_r}{N_{u}(N + L N_{r})}.
\end{equation}

\subsection*{Recovery at BBU}
Upon receiving the samples corresponding to $\mathbf{\hat{X}}$ and $\mathbf{\hat{H}}$ at the BBU, we reconstruct $\mathbf{Y_f}$ as
\begin{equation}
    \mathbf{\hat{Y}_f} = \mathbf{\hat{X} F_{L} \mathbf{\hat{H}}},
\end{equation}
for SU-MIMO, and as
\begin{equation}
    \mathbf{\hat{Y}_f} = \mathbf{\hat{X} \diag(F_{L}) \mathbf{\hat{H}}},
\end{equation}
for MU-MIMO. Then the standard pilot-based channel estimation can be used at the BBU for data recovery. We use the estimated channel coefficients to apply either maximal ratio combining (MRC) in the single user case or zero-forcing (ZF) equalization in the multi-user case to $\mathbf{\hat{Y}_f}$  in order to recover the user data $\mathbf{X_f}$. 

\par{We do not directly use the $\mathbf{\hat{X}}$ received from RRH as the estimate of the user data $\mathbf{X_f}$ for the following two reasons: 1) The solution \{$\mathbf{\hat{X}, \hat{H}}$\} obtained in Algorithms 1 and 2 are unique only up to a scalar constant, and the product in (9) or (10) can converge to $\mathbf{Y_f}$ even when $\mathbf{\hat{X}}$ and $\mathbf{\hat{H}}$ do not individually converge to the actual $\mathbf{X_f}$ and $\mathbf{H_t}$. A detailed discussion on this is given in the Appendix. 2) In the case of network architectures like Cloud Radio Access Network (C-RAN), where several RRHs share a single BBU pool, the BBU can have knowledge of interferers which allows it to get a better estimate of the channel than the $\mathbf{\hat{H}}$ received from the RRH.}

\section{Results and Discussion}
We compare the performance of the proposed
compression method against PCA compression in \cite{Choi2016} applied in the frequency domain based on three criteria: Compression Ratio (CR), Symbol Error Rate (SER) and algorithm complexity.

\subsection{Comparison of achieved Compression Ratios}
For a SU-MIMO system, the CR of the proposed method is given by (3). The CR for PCA compression in \cite{Choi2016} is given by
\begin{equation}
    CR_{SU,PCA} = \frac{NN_r}{L(N+N_{r})}.
\end{equation}

\begin{table}[b!]
\centering
\begin{tabular}{ |p{1cm}||p{2cm}|p{2cm}|  }\hline
\makebox[3em]{Method}&\makebox[5.5em]{$N=1024$}&\makebox[5.5em]{$N=4096$}\\\hline\hline
\makebox[3em]{MD}&\makebox[5.5em]{36.6}&\makebox[5.5em]{53.9}\\\hline
\makebox[3em]{PCA}&\makebox[5.5em]{5.0}&\makebox[5.5em]{5.2}\\\hline
\end{tabular}
\caption{Compression Ratios for Matrix Decomposition (MD) and PCA methods for $N_{r}=64, L=12,N_{u}=1$}
\end{table}

\begin{table}[b!]
\centering
\begin{tabular}{ |p{1cm}||p{2cm}|p{2cm}|  }\hline
\makebox[3em]{Method}&\makebox[5.5em]{$N=1024$}&\makebox[5.5em]{$N=4096$}\\\hline\hline
\makebox[3em]{MD}&\makebox[5.5em]{9.2}&\makebox[5.5em]{13.5}\\\hline
\makebox[3em]{PCA}&\makebox[5.5em]{1.2}&\makebox[5.5em]{1.3}\\\hline
\end{tabular}
\caption{Compression Ratios for MD and PCA methods for $N_{r}=64, L=12,N_{u}=4$}
\end{table}

We observe from (3) and (11) that for large values of $N$, which is the case when $N$ is the OFDM symbol length for large bandwidths, having $N \gg \max \{L,N_r\}$ implies
\begin{equation*}
    \frac{CR_{SU}}{CR_{SU,PCA}} = \frac{LN+LN_{r}}{N+LN_{r}} \approx \frac{LN}{N},
\end{equation*}
which gives us 
\begin{equation}
    CR_{SU} \approx L \times CR_{SU,PCA}.
\end{equation}
For the MU-MIMO case, 
\begin{equation}
    CR_{MU,PCA} = \frac{NN_r}{LN_{u}(N+N_{r})}.
\end{equation}
Comparing this with (8) when $N \gg \max \{L,N_r,N_u\}$ gives us
\begin{equation*}
    \frac{CR_{MU}}{CR_{MU,PCA}} = \frac{LN_{u}N+LN_{u}N_{r}}{N_{u}N+LN_{u}N_{r}} \approx L.
\end{equation*}
The CRs for different values of $N,N_{r},L,N_{u}$ for both the methods are given in Tables I and II. Thus, we observe that the proposed compression method can give nearly an order of magnitude higher CRs than PCA compression in \cite{Choi2016}, \cite{aswathylakshmi2019qr}.

\subsection{Symbol Error Rate Performance}
We use Monte Carlo simulations to evaluate the uncoded SER of the proposed method, PCA compression and the uncompressed system. The simulation parameters are summarized in Table III. We use the exponential correlation model \cite{loyka2001} with correlation coefficient 0.7 to model the antenna correlation for a uniform linear array at the RRH. We use $N_{u} = 4$ users for the MU-MIMO case.

\begin{table}[b!]
\centering
\begin{tabular}{ |p{5cm}||p{2cm}|  }
 \hline
 Modulation scheme ($M$) & 64-QAM \\\hline
 No. of RRH antennas ($N_r$) & 64\\\hline
 FFT size ($N$) & 4096\\\hline
 Multi-path Channel length ($L$) & 12\\\hline
 Channel Model & TDLA30 \\\hline
\end{tabular}
\caption{Simulation Parameters}
\end{table}

\begin{figure*}[t!]
  \centering
  \mbox{
    \subfigure[\label{subfigure label}]{\centering\includegraphics[width=.45\linewidth]{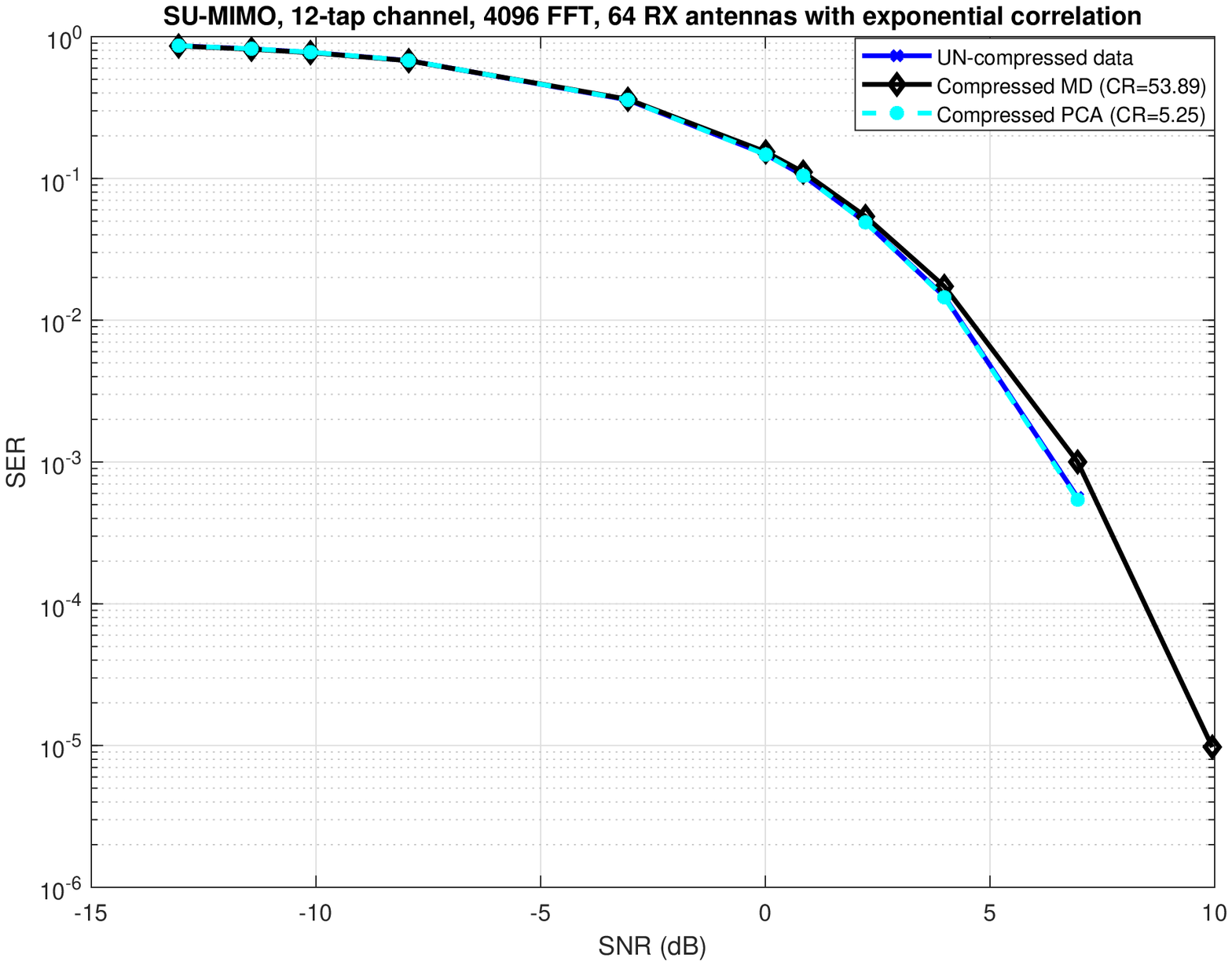}}\quad
    \subfigure[\label{subfigure label}]{\centering\includegraphics[width=.45\linewidth]{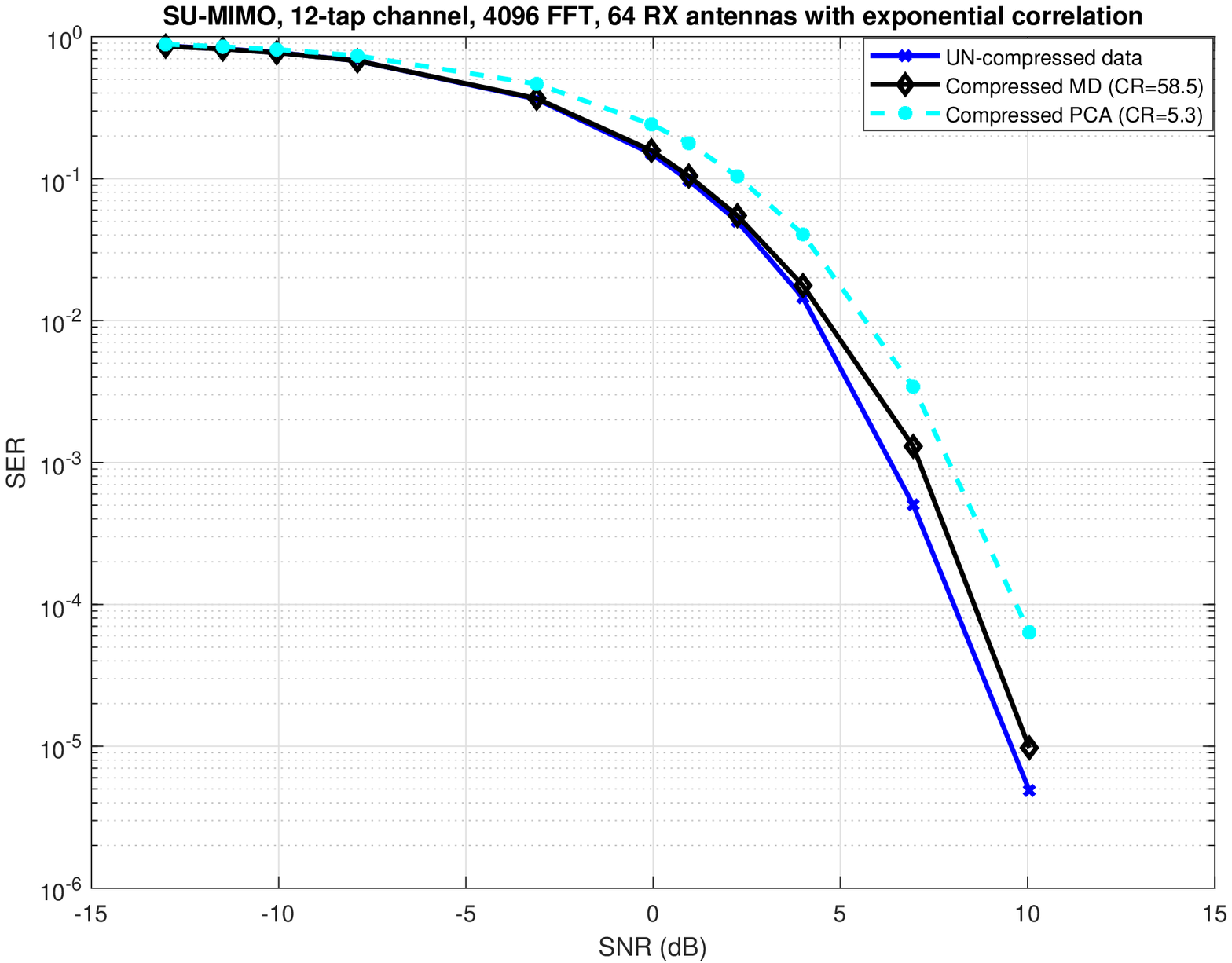}}\quad
 }
  \caption{Uncoded SERs of the proposed method (after 10 iterations of Algorithm 1) and PCA compression for single user in the system with (a) all $N_r=64$ antennas used during MRC (b) only $N_{r}/2=32$ antennas used during MRC. SER for the uncompressed system for all cases are also plotted as baseline. Proposed method matches the SER of uncompressed system in both cases while providing a CR above 50, whereas SER for PCA compression degrades in case (b) in spite of offering a CR of only about 5.}
  \label{main figure label}
\end{figure*}

\begin{figure*}[t!]
  \centering
  \mbox{
    \subfigure[\label{subfigure label}]{\centering\includegraphics[width=.45\linewidth]{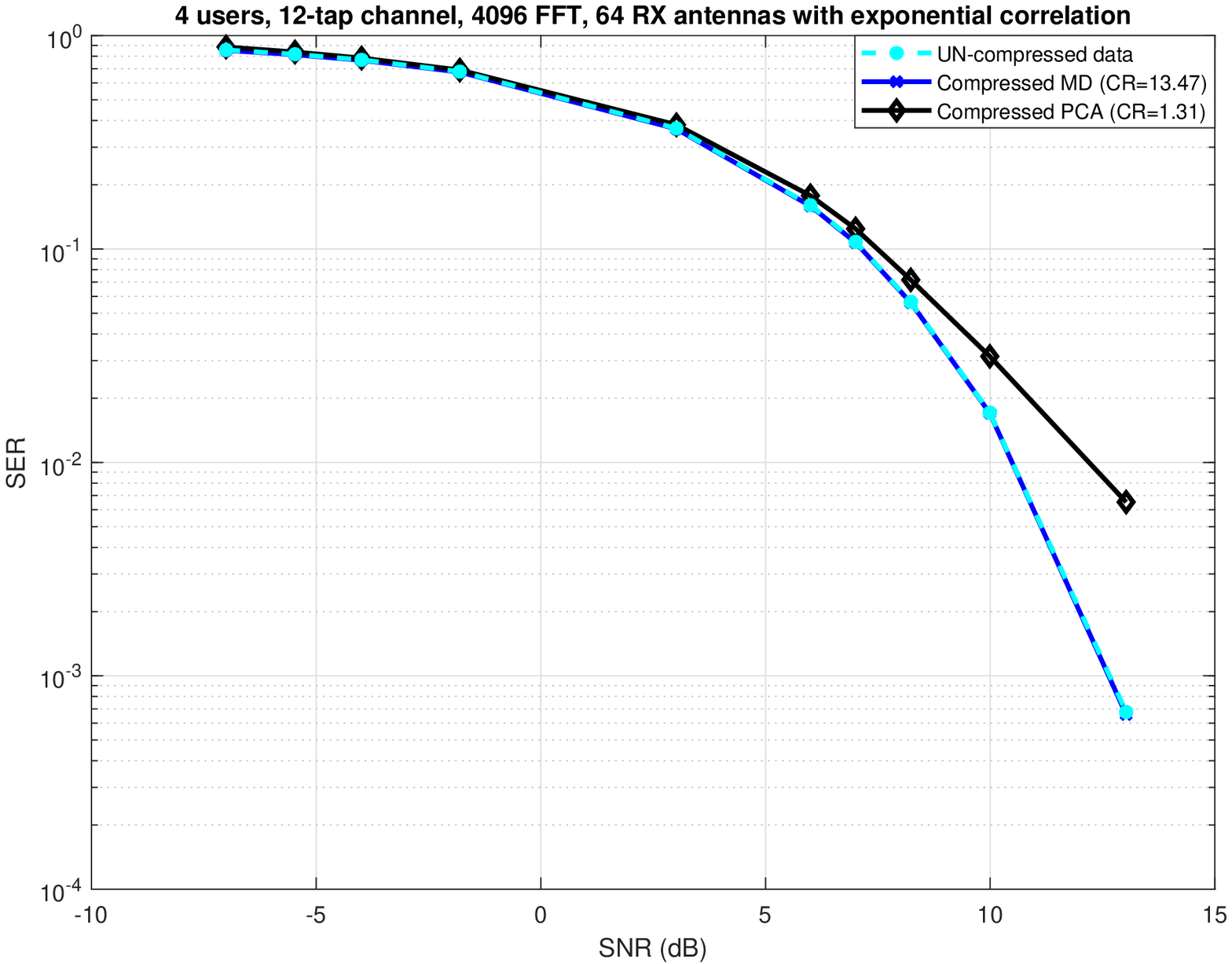}}\quad
    \subfigure[\label{subfigure label}]{\centering\includegraphics[width=.45\linewidth]{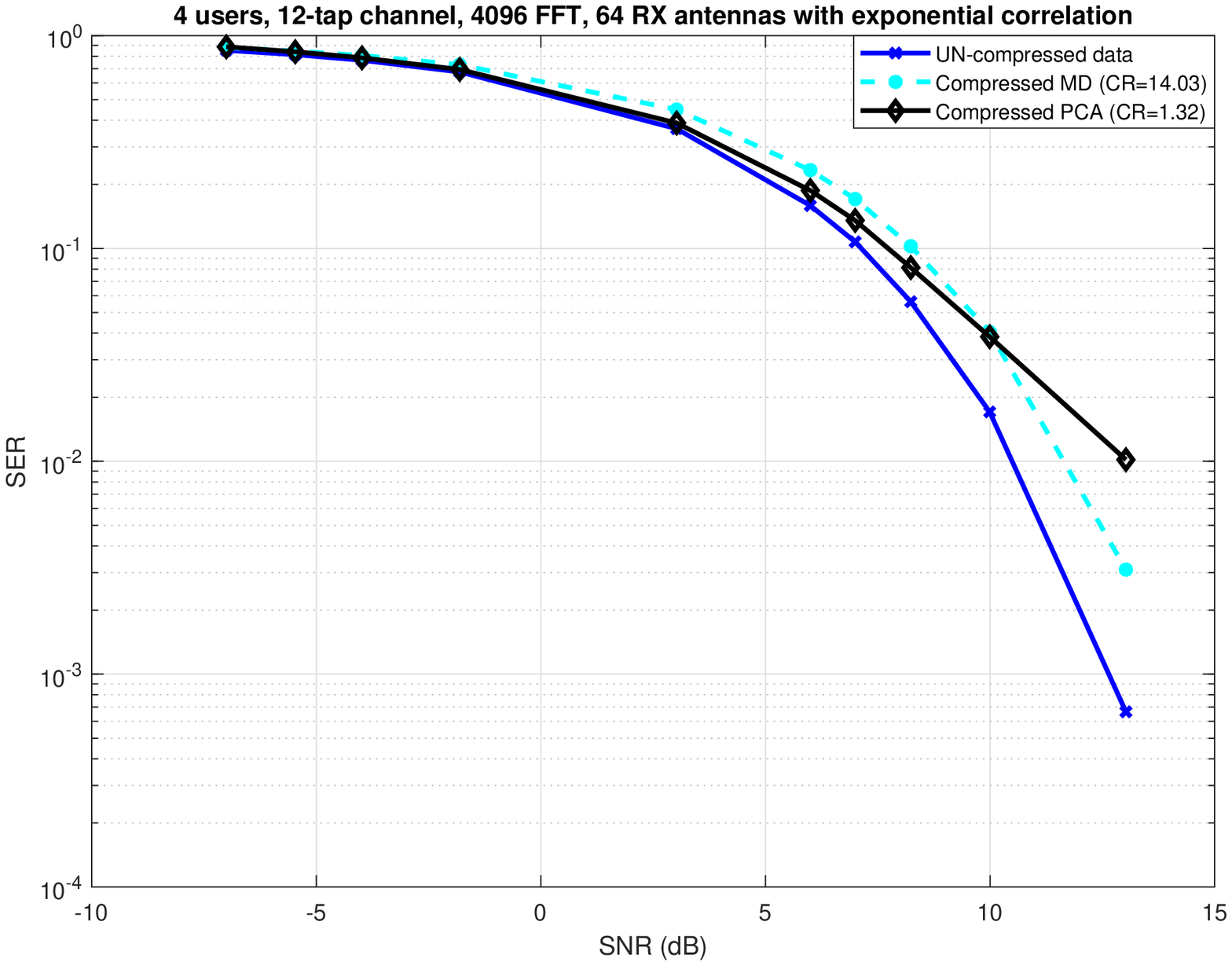}}\quad
}
  \caption{Uncoded SERs of the proposed method (after 10 iterations of Algorithm 2), PCA compression and no compression cases for 4 users in the system with (a) all $N_r=64$ antennas used during ZF (b) only $LN_{u}=48$ antennas used during ZF. In both cases, the proposed method offers CR of nearly 14 while matching the SER of PCA compression and uncompressed system at low SNRs. The SER of the proposed method degrades at higher SNRs due to loss of diversity.}
  \label{main figure label}
\end{figure*}

\par{Fig. 2 shows the uncoded SERs of the proposed method and PCA compression for a single user in the system. MRC using pilot-based channel estimates is applied at the BBU. In Fig. 2(a), all $N_r=64$ antennas are used during MRC. We observe that the proposed method performs as well as PCA and the non-compression case, even while providing a CR of 53.9. This is more than $10\times$ the CR provided by PCA. In Fig. 2(b), we halve the number of antennas used at the MRC stage for the compressed systems to increase the compression. We observe that the proposed method, with an improved CR of 58.5, still matches the performance of the uncompressed system whereas the performance of PCA compression degrades even though it provides a CR of only 5.3, which is $11\times$ lower than our method.}


\par{Fig. 3 shows the SER performance for the multi-user case with $N_{u}=4$ users in the system. In Fig. 3(a), data from all $N_r=64$ antennas are used for demodulation at BBU. In Fig. 3(b), data from only $LN_{u}=48$ antennas are used at the BBU to improve the CR. In both cases, the proposed method matches the SER of PCA and no compression at low SNRs, while providing a CR of approximately 14, nearly $10\times$ that of PCA. We observe loss of diversity gain at high SNRs for the proposed method in both cases and for PCA in case (b).}

\subsection{Algorithm Complexity}
We now analyze the time complexity of the proposed compression method. We observe from Algorithm 1 that the most computationally intensive operation is step 5, where we calculate 
\begin{equation}
    \mathbf{\hat{H}_{k+1}} = \mathbf{{(\hat{X}_{k}F_L)}^{\dagger}}\mathbf{Y_f} = (\mathbf{F_{L}^{H}\hat{X}_{k}^{H}\hat{X}_{k}F_{L}})^{-1}\mathbf{F_{L}^{H}\hat{X}_{k}^{H}Y_f}.
\end{equation}
The matrix pseudo-inverse in (14) is usually calculated via SVD, which has a time complexity of $\mathcal{O}(LN^{2})$ since $\mathbf{\hat{X}_{k}F_L}$ has dimension $N\times L$. However, if we expand the matrix pseudo-inverse, then the computations $(\mathbf{F_{L}^{H}\hat{X}_{k}^{H}\hat{X}_{k}F_{L}})^{-1}$ and $\mathbf{F_{L}^{H}\hat{X}_{k}^{H}Y_f}$ can be carried out in parallel to save time. Thus, we split (14) into three steps:
\begin{enumerate}
    \item Compute and store $\mathbf{F_{L}^{H}\hat{X}_{k}^{H}}$, which has $LN$ multiplications, same as a matrix-vector product because $\mathbf{\hat{X}_{k}^{H}}$ is diagonal. 
    \item Compute $(\mathbf{F_{L}^{H}\hat{X}_{k}^{H}\hat{X}_{k}F_{L}})^{-1}$ and $\mathbf{F_{L}^{H}\hat{X}_{k}^{H}Y_f}$ in parallel. $\mathbf{U} = \mathbf{F_{L}^{H}\hat{X}_{k}^{H}\hat{X}_{k}F_{L}}$ requires $L^{2}N$ multiplications and finding the inverse of $\mathbf{U}$, an $L\times L$ matrix, has complexity $\mathcal{O}(L^{3})$. Meanwhile, $\mathbf{V} = \mathbf{F_{L}^{H}\hat{X}_{k}^{H}Y_f}$ requires $LNN_r$ multiplications. Since $L<N_r$, this step has a time complexity of $\mathcal{O}(LNN_{r})$.
    \item Compute the product of $\mathbf{U}_{L\times L}^{-1}$ and $\mathbf{V}_{L\times N_r}$, which requires $L^{2}N_r$ multiplications.
\end{enumerate} 
Thus, the approximate time complexity of one iteration in Algorithm 1 is $\mathcal{O}(LNN_{r})$, which is comparable to the time complexity of $\mathcal{O}(N_{r}^3)$ of the SVD step in PCA compression \cite{Choi2016}. However, we require multiple iterations for Algorithm 1 to converge to $\mathbf{Y_f}$, which makes it computationally more intensive than PCA compression in \cite{Choi2016}. 

\section{Conclusion}
Fronthaul capacity is a major bottleneck in the implementation of 5G OFDM massive MIMO networks. In this work, we proposed a method of fronthaul compression for the uplink that exploits the convolution structure of the received data. We used an iterative alternating minimisation approach at the RRH to approximate the received signal as the product of a diagonal user data matrix and a low rank channel response matrix, allowing the received signal to be reconstructed at the BBU using fewer samples. The method can be tailored to both single-user and multi-user MIMO systems, and link level simulations show that it provides the same symbol error rates as an uncompressed system. It can offer nearly an order of magnitude higher compression ratios than the existing methods.

\section{Acknowledgements}
This work was supported by the Ministry of Electronics and Information Technology, India, through the 5G project, and by ANSYS SOFTWARE PVT. LTD. through their doctoral fellowship to Aswathylakshmi P.

\bibliographystyle{IEEEtrans}

\bibliography{ref}

\begin{thebibliography}{10}
\providecommand{\url}[1]{#1}
\csname url@samestyle\endcsname
\providecommand{\newblock}{\relax}
\providecommand{\bibinfo}[2]{#2}
\providecommand{\BIBentrySTDinterwordspacing}{\spaceskip=0pt\relax}
\providecommand{\BIBentryALTinterwordstretchfactor}{4}
\providecommand{\BIBentryALTinterwordspacing}{\spaceskip=\fontdimen2\font plus
\BIBentryALTinterwordstretchfactor\fontdimen3\font minus
  \fontdimen4\font\relax}
\providecommand{\BIBforeignlanguage}[2]{{%
\expandafter\ifx\csname l@#1\endcsname\relax
\typeout{** WARNING: IEEEtran.bst: No hyphenation pattern has been}%
\typeout{** loaded for the language `#1'. Using the pattern for}%
\typeout{** the default language instead.}%
\else
\language=\csname l@#1\endcsname
\fi
#2}}
\providecommand{\BIBdecl}{\relax}
\BIBdecl

\bibitem{larsson2017massive}
E.~G. Larsson and L.~Van~der Perre, ``Massive {MIMO} for 5{G},'' \emph{IEEE 5G
  Tech Focus}, vol.~1, no.~1, March 2017.

\bibitem{ecpri2017}
{\relax CPRI Consortium} \emph{et~al.}, ``e{CPRI} specification {V}1. 0,'' Aug
  2017.

\bibitem{guo2013lte}
B.~Guo, W.~Cao, A.~Tao, and D.~Samardzija, ``{LTE/LTE-A} signal compression on
  the {CPRI} interface,'' \emph{Bell Labs Technical Journal}, vol.~18, no.~2,
  pp. 117--133, 2013.

\bibitem{drvenica2016}
B.~Drvenica and G.~Luz, ``Compression analysis of massive {MIMO} uplink,''
  Master's thesis, Chalmers University of Technology, Gothenburg, Sweden, 2016.

\bibitem{Peng2016}
M.~{Peng}, Y.~{Sun}, X.~{Li}, Z.~{Mao}, and C.~{Wang}, ``Recent advances in
  cloud radio access networks: System architectures, key techniques, and open
  issues,'' \emph{IEEE Communications Surveys Tutorials}, vol.~18, no.~3, pp.
  2282--2308, thirdquarter 2016.

\bibitem{Choi2016}
J.~{Choi}, B.~L. {Evans}, and A.~{Gatherer}, ``Space-time fronthaul compression
  of complex baseband uplink {LTE} signals,'' \emph{IEEE International
  Conference on Communications (ICC)}, pp. 1--6, May 2016.

\bibitem{aswathylakshmi2019qr}
P.~Aswathylakshmi and R.~K. Ganti, ``{QR} approximation for fronthaul
  compression in uplink massive {MIMO},'' \emph{IEEE Globecom Workshops (GC
  Wkshps)}, pp. 1--7, 2019.

\bibitem{li2019rapid}
X.~Li, S.~Ling, T.~Strohmer, and K.~Wei, ``Rapid, robust, and reliable blind
  deconvolution via nonconvex optimization,'' \emph{Applied and computational
  harmonic analysis}, vol.~47, no.~3, pp. 893--934, 2019.

\bibitem{jain2013low}
P.~Jain, P.~Netrapalli, and S.~Sanghavi, ``Low-rank matrix completion using
  alternating minimization,'' \emph{Proceedings of the forty-fifth annual {ACM}
  symposium on Theory of computing}, pp. 665--674, 2013.

\bibitem{38101}
3GPP, ``{\relax 3GPP TS 38.101-1 V15.4.0, NR; User Equipment (UE) radio
  transmission and reception; Part 1: Range 1 Standalone (Release 15)},'' Tech.
  spec., Dec. 2018.

\bibitem{zaidi20185g}
A.~Zaidi, F.~Athley, J.~Medbo, U.~Gustavsson, G.~Durisi, and X.~Chen, ``{5G}
  physical layer: principles, models and technology components,''
  \emph{Academic Press}, 2018.

\bibitem{alliance2020ran}
{ORAN}-Alliance, ``{O-RAN} use cases and deployment scenarios,'' \emph{White
  Paper}, 2020.

\bibitem{loyka2001}
S.~L. Loyka, ``Channel capacity of {MIMO} architecture using the exponential
  correlation matrix,'' \emph{IEEE Communications letters}, vol.~5, no.~9, pp.
  369--371, 2001.

\end{thebibliography}

\appendix
\section*{Uniqueness of the solution}
The optimal solution to (4), \{$\mathbf{\hat{X}, \hat{H}}$\} obtained via Algorithm 1 is unique up to a scalar constant. We prove this in the following lemma.

\begin{lemma}
    If \{$\mathbf{\hat{X}, \hat{H}}$\} and \{$\mathbf{\tilde{X}, \tilde{H}}$\} are two solutions to (4), then they are related to each other by the scalar transform
    \begin{equation}
        \mathbf{\tilde{X}} = (\tfrac{1}{\lambda})\mathbf{\hat{X}} \hspace{2mm}\text{and}\hspace{2mm} \mathbf{\tilde{H}} = \lambda\mathbf{\hat{H}}.
    \end{equation}
\end{lemma}
\begin{proof}
    If \{$\mathbf{\hat{X}, \hat{H}}$\} and \{$\mathbf{\tilde{X}, \tilde{H}}$\} are both solutions to (4), then
    \begin{equation}
        \mathbf{\hat{X}F_{L}\hat{H}} = \mathbf{\tilde{X}F_{L}\tilde{H}}.
    \end{equation}
    Let $\mathbf{\tilde{X}} = \mathbf{\hat{X}X^{*}}$ and $\mathbf{\tilde{H}} = \mathbf{H^{*}\hat{H}}$, where $\mathbf{X^*}$ is an $N\times N$ matrix and $\mathbf{H^*}$ is an $L\times L$ matrix. Then, we need  
    \begin{equation}
        \mathbf{X^{*}F_{L}H^{*}} = \mathbf{F_L}
    \end{equation}
    to satisfy (16). If $\mathbf{X^*}$ is a rank-$N$ diagonal matrix, then $\mathbf{F_{L}H^*} = \mathbf{(X^{*})^{-1}F_L}$.
    Using $\mathbf{f_{m}}$ to denote the $m^{th}$ row of $\mathbf{F_L}$ and $\lambda_{m}$ to denote the $m^{th}$ diagonal element of $\mathbf{(X^{*})^{-1}}$, we can rewrite the above as
    \begin{equation}
        \mathbf{f_{m}H^*} = \lambda_{m}\mathbf{f_{m}}, \hspace{2mm} m = 1,2,...,N.
    \end{equation}
    Thus, the rows of $\mathbf{F_L}$ form the left eigen vectors of the matrix $\mathbf{H^*}$. $\mathbf{F_L}$ is a rank-$L$ matrix, therefore we can express $\mathbf{f_{L+1}}$ as
    \begin{equation}
        \mathbf{f_{L+1}} = \sum_{i=1}^{L}a_{i}\mathbf{f_i}.
    \end{equation}
    Combining (18) and (19), we have
    \begin{equation*}
        \mathbf{f_{L+1}H^*} = \sum_{i=1}^{L}a_{i}\lambda_{i}\mathbf{f_{i}} = \lambda_{L+1}\sum_{i=1}^{L}a_{i}\mathbf{f_{i}},
    \end{equation*}
    which can hold true only when all $\lambda_{i}$'s are equal. This gives us 
    \begin{equation}
        \mathbf{X^*} = (\tfrac{1}{\lambda})\mathbf{I_{N}} \hspace{1mm}\text{and}\hspace{1mm} \mathbf{H^*} = \lambda\mathbf{I_{L}},
    \end{equation}
    where $\lambda_{i} = \lambda$, for $i=1,2,...,N$ and $\mathbf{I_{K}}$ denotes identity matrix of dimension $K$. 

    \par{We note that \{$\mathbf{\hat{X}X^{*}, \hat{H}H^{*}}$\} is also a solution to (4) where $\mathbf{H^{*}} = \lambda \mathbf{I_{N_r}}$. To see this, use $\mathbf{A} = \mathbf{F_{L}\hat{H}}$} in (16). Then the condition to be satisfied, $\mathbf{X^{*}AH^{*}} = \mathbf{A}$ is of the form in (17) and the result follows. 
\end{proof}


\end{document}